\documentclass[preprint,12pt]{elsarticle}

\usepackage{amsthm,amsmath,amssymb}
\usepackage{algorithm,algpseudocode}
\usepackage[unicode]{hyperref}
\usepackage{booktabs}
\usepackage{comment}
\usepackage{geometry}
  
\theoremstyle{definition}
\newtheorem{theorem}              {Theorem}
\newtheorem{lemma}      [theorem] {Lemma}

\newtheorem{corollary}  [theorem] {Corollary}

\newtheorem{problem}    [theorem] {Problem}

\numberwithin{equation}{section}
\numberwithin{figure}{section}
\numberwithin{table}{section}

\newcommand{\argmax}{\operatornamewithlimits{argmax}}

\makeatletter
\def\ps@pprintTitle{%
 \let\@oddhead\@empty
 \let\@evenhead\@empty
 \def\@oddfoot{\centerline{\thepage}}%
 \let\@evenfoot\@oddfoot}
\makeatother

\usepackage{color}
\usepackage{CJKutf8}
\usepackage{ifthen}
\newcommand{\COMM}[2]{{
\begin{CJK}{UTF8}{ipxm}
\ifthenelse{\equal{#1}{TM}}{\color{blue}}{
\ifthenelse{\equal{#1}{KA}}{\color{red}}{
\ifthenelse{\equal{#1}{AA}}{\color{cyan}}{
\ifthenelse{\equal{#1}{BB}}{\color{magenta}}}}}
[#1: #2]
\end{CJK}
}}

\begin{document}

\begin{frontmatter}

\title{Optimal Algorithm to Reconstruct a Tree from a Subtree Distance}

\author{Takanori Maehara}
\address{RIKEN Center for Advanced Intelligence Project \\ takanori.maehara@riken.jp}

\author{Kazutoshi Ando}
\address{Department of Mathematical and Systems Engineering, Shizuoka University \\ ando.kazutoshi@shizuoka.ac.jp}

\begin{abstract}
This paper addresses the problem of finding a representation of a subtree 
distance, which is an extension of the tree metric. 
We show that a minimal representation is uniquely determined by a given subtree distance, and give a linear time algorithm that finds such a representation. 
This algorithm achieves the optimal time complexity. 
\end{abstract}


\end{frontmatter}

\section{Introduction and the results}

A \emph{phylogenetic tree} represents an evolutionary relationship among the species that are investigated.
Estimating a phylogenetic tree from experimental data is a fundamental problem in phylogenetics~\cite{wiley2011phylogenetics}.
One of the commonly used approaches to achieve this task is the use of a \emph{distance-based method}.
In this approach, we first compute the dissimilarity (i.e., a nonnegative and symmetric function) between the species by, e.g., the edit distance between the genome sequences.
Then, we find a weighted tree having the shortest path distance that best fits the given dissimilarity.
The most popular method for this approach is the neighbor-joining method~\cite{saitou1987neighbor}. 

A weighted tree $\mathcal{T}$ is specified by the set of vertices $\mathcal{V}(\mathcal{T})$, the set of edges $\mathcal{E}(\mathcal{T})$, and the nonnegative edge weight $l:\mathcal{E}(\mathcal{T})\to\mathbb{R}_+$.
Let us consider the case in which a given dissimilarity $d: X\times X\to \mathbb{R}$ exactly fits some weighted tree; i.e., there exists a weighted tree $\mathcal{T}$ and a mapping $\psi: X\to \mathcal{V}(\mathcal{T})$ such that
\begin{align}
  \label{eq:exactly_fits} 
  d(x,y)=d_\mathcal{T}(\psi(x),\psi(y)) \quad(x,y\in X),
\end{align}
where $d_\mathcal{T}(u,v)$ is the distance between $u$ and $v$ in $\mathcal{T}$ for $u, v \in \mathcal{V}(\mathcal{T})$.
In this case, the dissimilarity $d:X\times X\to \mathbb{R}$ is called a \emph{tree metric}, and the pair $(\mathcal{T},\psi)$ is called a \emph{representation} of $d$.
It is known that a dissimilarity $d$ is a tree metric if and only if it satisfies an inequality called the \emph{four-point condition}~\cite{zaretskii1965constructing,buneman1971recovery}, which is given by
\begin{equation}
  d(x,y)+d(z,w)\leq\max\{d(x,z)+d(y,w),d(x,w)+d(y,z)\} 
\end{equation}
for any $x, y, z, w \in X$.
For any tree metric $d$, there exists a unique minimal representation $(\mathcal{T},\psi)$ of $d$ \cite{buneman1971recovery}.
Here, a representation is \emph{minimal} if there is no representation $(\mathcal{T}',\psi')$ of $d$, such that $\mathcal{T}'$ is obtained by removing some vertices and edges and/or by contracting some edges of $\mathcal{T}$ (i.e., $\mathcal{T}'$ is a proper minor of $\mathcal{T}$).
Furthermore, such a representation is constructed in $O(n^2)$ time~\cite{culberson1989fast}, where $n=|X|$. 

In some applications, we are interested in the distance between \emph{groups of species} (e.g., genus, tribe, or family). 
In such a case, we aim to identify a group as a connected subgraph in a phylogenetic tree. 
The subtree distance, which was introduced by Hirai~\cite{hirai2006characterization}, is an extension of the tree metric that can be adopted for use in such situations. 
A function $d:X\times X\to\mathbb{R}$ is called a \emph{subtree distance} if there exists a weighted tree $\mathcal{T}$ and a mapping $\phi: X\to 2^{\mathcal{V}(\mathcal{T})}$ such that $\phi(x)$ induces a subtree of $\mathcal{T}$ (i.e., a connected subgraph of $\mathcal{T}$) for $x\in X$ and equations 
\begin{align}
  \label{eq:exactly_fits2} 
  d(x,y)=d_\mathcal{T}(\phi(x),\phi(y)) \quad(x,y\in X),
\end{align}
hold, where $d_\mathcal{T}(U,W)=\min\{d_\mathcal{T}(u,w)\mid u\in U, w\in W\}$ for $U,W\subseteq\mathcal{V}(\mathcal{T})$.
We say that a pair $(\mathcal{T},\phi)$ is a \emph{representation} of $d$.
Note that a subtree distance is not necessarily a metric because it may not satisfy the non-degeneracy ($d(x,y)>0$ for $x \neq y$) and the triangle inequality ($d(x,z)\leq d(x,y)+d(y,z)$). 
Hirai proposed a characterization of subtree distances in which a dissimilarity $d:X\times X\to \mathbb{R}$ is a subtree distance if and only if it satisfies an inequality called the \emph{extended four-point condition}, which is given by 
\begin{align}
  \label{eq:4pc'}
  d(x,y)+d(z,w) \leq \max\left\{
  \begin{array}{l}
  d(x,z)+d(y,w),d(x,w)+d(y,z),\\
  d(x,y),d(z,w),\\
  \frac{d(x,y)+d(y,z)+d(z,x)}{2},\frac{d(x,y)+d(y,w)+d(w,x)}{2},\\
  \frac{d(x,z)+d(z,w)+d(w,x)}{2},\frac{d(y,z)+d(z,w)+d(w,y)}{2}
  \end{array}
\right\}
\end{align}
for any $x,y,z,w\in X$.
The extended four-point condition yields an $O(n^4)$ time algorithm to recognize a subtree distance. 

This paper addresses the following problem, which we call \emph{the subtree distance reconstruction problem}. 
\begin{problem}
Given a subtree distance $d: X\times X\to\mathbb{R}$ on a finite set $X$, find a representation $(\mathcal{T},\phi)$ of $d$. 
\end{problem}
\noindent Ando and Sato~\cite{ando2017algorithm} proposed an $O(n^3)$ time algorithm for this problem.
Their algorithm consists of three steps: 
(1) identify a subset $V_0=\{x\in X\mid d(y,z)\leq d(x,y)+d(x,z)\ (y,z\in X)\}$, 
(2) find a representation $(\mathcal{T},\phi)$ for the restriction of $d$ onto $V_0$, and 
(3) for $x\in X\setminus V_0$, locate $\phi(x)$ in $\mathcal{T}$ by examining a connected components of $\mathcal{T}\setminus \phi(x)$. 

In this study, we propose the following theorems. 
We define a minimal representation for a subtree distance in the 
same manner as a minimal representation for a tree metric. 
\begin{theorem}\label{th:main_theorem}
For a subtree distance $d:X\times X\to \mathbb{R}$, a minimal representation 
$(\mathcal{T},\phi)$ is uniquely determined by $d$. 
\end{theorem}
\begin{theorem}\label{th:main_theorem_alg}
There exists an $O(n^2)$ time algorithm that finds, for any subtree distance 
$d: X \times X \to \mathbb{R}$, its unique minimal representation, 
where $n = |X|$.
\end{theorem}

The proof of Theorem~\ref{th:main_theorem_alg} is constructive. 
Similar to~\cite{ando2017algorithm}, our algorithm consists of three parts: 
(1) identify the set of objects $L$ that are \emph{mapped to the leaves}, 
(2) find the minimal representation $(\mathcal{T},\phi)$ for the restriction of $d$ onto $L$, and 
(3) for $x\in X\setminus L$, locate $\phi(x)$ in ${\mathcal T}$ by \emph{measuring the distances from the leaves}.
Since Steps 1 and 3 can be implemented with a time complexity of $O(n^2)$, and there is an $O(n^2)$ time algorithm for Step 2~\cite{culberson1989fast}, the total time complexity of the algorithm is $O(n^2)$. 
Note that even if we know $d \colon X \times X \to \mathbb{R}$ is a tree metric, $\Omega(n^2)$ time is required to reconstruct a tree~\cite{hein1989optimal}.
Therefore, our algorithm achieves the optimal time complexity.
 
This algorithm can also be used to recognize a subtree distance by checking the failure or inconsistency during the process and by verifying equations \eqref{eq:exactly_fits2} after the reconstruction.
\begin{corollary}\label{co:main_theorem_alg}
There exists an $O(n^2)$ time algorithm that determines whether a given input $d: X \times X \to \mathbb{R}$ is a subtree distance or not, where $n = |X|$.
\end{corollary}



\section{Proofs}

We assume that there are no objects $x, y \in X$ such that $d(x,z) = d(y,z)$ for all $z \in X$. 
This assumption is satisfied by removing such elements after lexicographic sorting, which requires $O(n^2)$ time~\cite{wiedermann1979complexity}. 
Clearly, this preprocessing does not change the minimal representation.
We also assume that $|X| \ge 3$. Otherwise, the theorems trivially hold.

First, we prove Theorem~\ref{th:main_theorem}.
We identify the properties of a minimal representation.
\begin{lemma}
\label{lem:mimimal_representation}
Let $(\mathcal{T}, \phi)$ be a minimal representation of a subtree distance $d: X \times X \to \mathbb{R}$.
Then, the following properties hold.
\begin{enumerate}
  \item For each edge $e \in \mathcal{E}(\mathcal{T})$, the length of $e$ is positive.
  \item For each leaf vertex $u \in \mathcal{V}(\mathcal{T})$, there exists $x \in X$ such that $\phi(x) = \{u\}$.
\end{enumerate}
\end{lemma}
\begin{proof} \ 
1. If there is an edge $e$ of zero length, we can contract $e$ from the representation.

2. Let $u$ be a leaf vertex of $\mathcal{T}$. 
If there is no $x \in X$ with $u \in \phi(x)$, we can remove $u$ from the representation to obtain a smaller representation.
Suppose that, for all $x \in X$ with $u \in \phi(x)$, $\phi(x)$ contains at least two elements.
Then, these $\phi(x)$s contain the unique adjacent vertex $v$ of $u$. 
Since $d(v, w) \le d(u, w)$ for all $w \in \mathcal{V}(\mathcal{T}) \setminus \{u\}$, $u$ does not contribute any shortest paths in the tree.
Therefore, we can remove $u$ from the representation.
\qedhere
\end{proof}

Motivated by Property~2 in Lemma~\ref{lem:mimimal_representation}, we introduce the following definition.
For a minimal representation $(\mathcal{T}, \phi)$, an object $x \in X$ is a \emph{leaf object} if $\phi(x) = \{u\}$ for some leaf $u \in \mathcal{V}(\mathcal{T})$.

We defined a leaf object by specifying a minimal representation. However, as shown below, the set of leaf objects is uniquely determined by $d$.
First, we show that there exists an object that is a leaf object for any minimal representation.
\begin{lemma}
\label{lem:farthest}
Let $(r, r') \in \argmax_{x,y \in X} d(x,y)$. 
Then, for any minimal representation, $r$ and $r'$ are leaf objects.
\end{lemma}
\begin{proof}
For any tree, the farthest pair is attained by a pair of leaves.
\end{proof}
Next, we show that the leaf objects are characterized by $d$ and a leaf object $r$.
\begin{lemma} \label{lem:leaf_characterization}
Let $(\mathcal{T}, \phi)$ be a minimal representation of subtree distance $d: X \times X \to \mathbb{R}$, and let $r \in X$ be a leaf object.
An object $x \in X \setminus \{r\}$ is a leaf object if and only if $d(y,r) < d(x,y) + d(x,r)$ for all $y \in X \setminus \{x,r\}$.
\end{lemma}
\begin{proof}
(The ``if'' part).
Suppose $x$ is not a leaf object. 
Then, there is another leaf object $y \in X \setminus \{x, r\}$ such that the path from $\phi(r)$ to $\phi(y)$ intersects $\phi(x)$.
By considering distances among $\phi(r)$, $\phi(y)$ and $\phi(x)$ on this path, we have  $d(r,y) \ge d(x,r) + d(x,y)$.

(The ``only if'' part). Suppose $x$ is a leaf object.
Then, for any $y \in X \setminus \{x,r\}$ the shortest path from $\phi(y)$ to $\phi(r)$ never intersects $\phi(x)$. 
Thus, we have $d(y,r) < d(x,y) + d(x,r)$.
Here, we used the assumption that there is no $x, y \in X$ such that $d(x,z) = d(y,z)$ for all $z \in X$.
\end{proof}

Since we can take a leaf object $r$ universally by Lemma~\ref{lem:farthest}, and the condition in Lemma~\ref{lem:leaf_characterization} is described without specifying the underlying representation, we can conclude that the set of leaf objects is uniquely determined by $d$.
Thus, we obtain the following corollary.
\begin{corollary}
Any minimal representation of a subtree distance has the same leaf objects. 
\qed
\end{corollary}

Now, we observe that the dissimilarity $d|_L: L \times L \to \mathbb{R}$ obtained by restricting $d: X \times X \to \mathbb{R}$ on the set of leaf objects $L \subseteq X$ forms a tree metric because the leaf objects are mapped to singletons.
Since the minimal representation of a tree metric is unique~\cite{buneman1971recovery}, any minimal representation of $d$ has the same topology as the minimal representation of $d|_L$.

The remaining issue is to show that each non-leaf object $x \in X \setminus L$ is uniquely located in the minimal representation $(\mathcal{T}, \phi)$ of $d|_L$.
This is clear because any connected subgraph in a tree is uniquely identified by the distances from the leaves.
More precisely, we obtain the following explicit representation.

We first consider $\mathcal{T}$ as a continuous object.
We fix a leaf object $r \in L$. 
For each leaf object $x \in L \setminus \{r\}$, there exists a unique path $\text{path}(\phi(r), \phi(x))$ from $\phi(r)$ to $\phi(x)$ in $\mathcal{T}$.
Let $I(r,a;x,b)$ be the interval on the path having a distance of at least $a$ from $\phi(r)$ and at least $b$ from $\phi(x)$, i.e.,
\begin{align}
	I(r,a;x,b) = \{ u \in \text{path}(\phi(r), \phi(x)) : d_\mathcal{T}(\phi(r), u) \ge a, d_\mathcal{T}(\phi(x), u) \ge b \}.
\end{align}
For $U \subseteq \mathcal{V}(\mathcal{T})$, we denote by $\overline{U}$ the subgraph of $\mathcal{T}$ induced by $U$.
Note that both $I(r, a; x, b)$ and $\overline{U}$ are continuous objects.
By using these notations, we obtain the following.
\begin{lemma}
\label{lem:nonleaf}
Let $d: X \times X \to \mathbb{R}$ be a subtree distance and $L$ be the set of leaf objects.
Let $(\mathcal{T}, \phi)$ be the minimal representation of $d|_L$. 
Fix a leaf object $r \in L$.
Then, we have for each non-leaf object $z \in X \setminus L$
\begin{align}
\label{eq:phiz}
	\overline{\phi(z)} = \bigcup_{x \in L \setminus \{r\}} I(r, d(r,z); x, d(x,z)).
\end{align}
\end{lemma}
\begin{proof}
Since $\overline{\phi(z)}$ is a connected subgraph, the intersection of $\overline{\phi(z)}$ and the path from $\phi(r)$ to $\phi(x)$ is the interval $I(r, d(r,z); x, d(x,z))$.
  Since any tree is covered by the paths from a fixed leaf $\phi(r)$ and the other leaves $\phi(x)$ for $x \in L \setminus \{r\}$, we have
\begin{align}
  \overline{\phi(z)} &= \left( \bigcup_{x \in L \setminus \{r\}} \text{path}(\phi(r), \phi(x)) \right) \cap \overline{\phi(z)} \notag \\
  &= \bigcup_{x \in L \setminus \{r\}} \left( \text{path}(\phi(r), \phi(x)) \cap \overline{\phi(z)} \right) \notag \\
  &= \bigcup_{x \in L \setminus \{r\}} I(r, d(r,z); x, d(x,z)). 
\end{align}
\end{proof}
By placing the vertices on the boundaries of $\overline{\phi(z)}$ for $z \in X \setminus L$ and then letting $\phi(z)$ be the vertices intersecting $\overline{\phi(z)}$ for $z \in X \setminus L$, we obtain the minimal representation of $d$.
This completes the proof of Theorem~\ref{th:main_theorem}.
Note that the number of the vertices of the minimal representation is $O(n^2)$ since each $\overline{\phi(z)}$ has at most $|L| = O(n)$ boundaries.

Now, we prove Theorem~\ref{th:main_theorem_alg}.
The above proof of Theorem~\ref{th:main_theorem} is constructive, and it provides the following algorithm:
\begin{enumerate}
\item Identify the set $L$ of leaf objects by Lemmas~\ref{lem:farthest} and \ref{lem:leaf_characterization}.
\item Find the minimal representation of $d|_L$ by the existing algorithm.
\item Locate the non-leaf objects by Lemma~\ref{lem:nonleaf}.
\end{enumerate}
We evaluate the time complexity of this algorithm.
Step 1 is conducted in $O(n^2)$ time for finding a leaf object $r \in X$ and $O(n^2)$ time for finding other leaf objects.
Step 2 is performed in $O(n^2)$ time by using Culberson and Rudnicki's algorithm~\cite{culberson1989fast}.
Also, Step 3 is performed in $O(n^2)$ time by equation \eqref{eq:phiz} since, for each $z$, it processes each edge at most once.
Hence, Theorem~\ref{th:main_theorem_alg} is proved.

\section*{Acknowledgments}

The authors thank Hiroshi Hirai and anonymous referees for helpful comments.
This work was supported by the Japan Society for the Promotion of Science, KAKENHI Grant Number 15K00033.

\bibliographystyle{plain}
\bibliography{main}

\end{document}